\documentclass[11pt]{article}
\usepackage{amsmath,bm,amsfonts,color,amsthm}
\usepackage[compress, sort]{natbib}
\bibpunct{(}{)}{;}{a}{}{,}
	
\RequirePackage[colorlinks,citecolor=blue,urlcolor=blue]{hyperref}
\usepackage{graphicx}

\DeclareMathOperator*{\argmin}{argmin} 

\title{Smoothing, Clustering, and Benchmarking for Small Area Estimation}
\author{Rebecca C. Steorts,  Carnegie Mellon University}
\begin{document}
\newcounter{itemnum}
\newtheorem{lemma}{Lemma}
\newtheorem{theorem}{Theorem}
\newtheorem{corollary}{Corollary}

\theoremstyle{definition}
\newtheorem{definition}{Definition}
\newtheorem*{remark}{Remark}

\newtheoremstyle{lemma}
{\topsep} 
{\topsep} 
{\it} 
{} 
{\bf} 
{:} 
{0.5em} 
{} 
\theoremstyle{lemma}
\newtheorem{lem}{Lemma}

\newcommand{\la}{\lambda}
\newcommand{\bl}{\bm{\lambda}}
\newcommand{\bsee}{\bm{\Sigma_e}}
\newcommand{\bo}{\bm{\Omega}}
\newcommand{\bD}{\bm{\Delta}}
\newcommand{\bT}{\bm{t}}
\newcommand{\bc}{\bm{C}}
\newcommand{\bww}{\bm{W}}
\newcommand{\bbb}{\bm{B}}
\newcommand{\bg}{\bm{G}}
\newcommand{\bmy}{\bm{y}}
\newcommand{\bdd}{\bm{\delta}}
\newcommand{\bdh}{\hat{\bm{\delta}}}

\newcommand{\thatb}{\hat{\bm{\theta}}^B}
\newcommand{\thatt}{\hat{\bm{\theta}}}
\newcommand{\thattt}{\hat{\bm{\theta_2}}}
\newcommand{\thp}{\hat{\bm{\theta}}^p}
\newcommand{\thg}{\hat{\bm{\theta}}^g}
\newcommand{\thbm}{\hat{\bm{\theta}}^{BM}}

\newcommand{\bty}{\tilde{\bm{\theta}}_{y,t}}
\newcommand{\su}{{\sigma_u^2}}
\newcommand{\thh}{{\theta}}
\newcommand{\bgam}{\bm{\gamma}}
\newcommand{\bth}{\bm{\theta}}
\newcommand{\bxi}{\bm{x}_i}
\newcommand{\bse}{\bm{\Sigma_e}^{-1}}
\newcommand{\bh}{\bm{H}}
\newcommand{\bhh}{\bm{h}}
\newcommand{\bd}{\bm{D}}
\newcommand{\ba}{\bm{A}}
\newcommand{\iprime}{{i^\prime}}
\newcommand{\jprime}{{j^\prime}}
\newcommand{\bb}{\bm{\beta}}
\newcommand{\bys}{\bm{y^*}}
\newcommand{\byst}{\bm{y^{*T}}}
\newcommand{\V}       {\text{Var}}
\newcommand{\de}       {\mbox{$\delta_i$}}
\newcommand{\tij}       {\mbox{$\theta_{ij}$}}
\newcommand{\htijb}       {\mbox{$\hat{\theta}_{ij}^B$}}

\newcommand{\thbariw}     {\mbox{$\bar{\hat{\theta}}_{iw}^{B}$}}
\newcommand{\tbw}       {\mbox{$\bar{\theta}_{w}^B$}}
\newcommand{\tbiw}       {\mbox{$\bar{\theta}_{iw}$}}
\newcommand{\thwb}       {\mbox{$\bar{\hat\bm{{\theta}}}_{w}^B$}}
\newcommand{\thw}       {\mbox{$\bar{\hat{\theta}}_{w}^B$}}
\newcommand{\thbiw}       {\mbox{$\bar{\hat{\theta}}_{iw}^B$}}
\newcommand{\that}       {\mbox{$\hat{\theta}$}}
\newcommand{\thij}       {\mbox{$\hat{\theta}_{ij}$}}
\newcommand{\thhi}       {\mbox{$\hat{\theta}_{i}$}}
\newcommand{\thi}       {\theta_i}

\newcommand{\lao}       {\mbox{$\lambda_{1i}$}}
\newcommand{\lat}       {\mbox{$\lambda_{2}$}}
\newcommand{\latt}       {\mbox{$\lambda_{3i}$}}

\newcommand{\cd}{\buildrel d \over \longrightarrow}
\newcommand{\cp}{\buildrel P \over \longrightarrow}
\newcommand{\hatbb}{\boldsymbol{\hat{b}}}
\newcommand{\hatbB}{\boldsymbol{\hat{B}}}
\newcommand{\hatbd}{\boldsymbol{\hat{d}}}
\newcommand{\commentt}[1]{}
\newcommand{\myvfil}[1]{\vskip 0pt plus #1fill}
\newcommand{\lik}{\ell_y(\theta)}
\newcommand{\likil}{\ell_y(\theta_i^{(l)})}
\newcommand{\hd}{\hfill$\diamondsuit$}
\newcommand{\aaa}{\epsilon}
\newcommand{\lt}{\left}
\newcommand{\rt}{\right}
\newcommand{\mbi}{\max_{1 \leq i \leq m} \bxi'\bb}
\newcommand{\bs}{B_{i*}}
\newcommand{\D}{\Delta}
\newcommand{\bi}{B_{i}}
\newcommand{\utheta}        {\mbox{$\boldsymbol{\theta}$}}
\newcommand{\thhj}{\hat{\theta}_j}
\newcommand{\thhij}{\hat{\theta}_{ij}}
\newcommand{\thiHB}{\hat{\theta}_i^{HB}}
\newcommand{\thih}{\hat{\theta}_i^H}
\newcommand{\thit}{\tilde{\theta}_i^H}

\newcommand{\tr}{\text{tr}}
\newcommand{\btt}{\boldsymbol{\theta}}

\newcommand{\ttt}{\boldsymbol{t}}
\newcommand{\bhat}{\boldsymbol{\hat{\beta}}}
\newcommand{\thb}{\bar{\theta}}
\newcommand{\bx}{\boldsymbol{x}}
\newcommand{\pxv}{{P}_X^V}
\newcommand{\pxvt}{\bar{P}_x^{V'}}
\newcommand{\pxvs}{\bar{P}_x^{V_*}}
\newcommand{\bv}{\boldsymbol{v}}
\newcommand{\bu}{\boldsymbol{u}}
\newcommand{\ur}{\boldsymbol{r}}
\newcommand{\uphi}{\boldsymbol{\phi}}
\newcommand{\uone}{\boldsymbol{1}}
\newcommand{\ue}{\boldsymbol{e}}
\newcommand{\uc}{\boldsymbol{c}}
\newcommand{\bbi}{\boldsymbol{b}_i}
\newcommand{\uw}{\boldsymbol{w}}
\newcommand{\bz}{\boldsymbol{z}}
\newcommand{\be}{\boldsymbol{e}}
\newcommand{\by}{\boldsymbol{y}}
\newcommand{\utt}{\boldsymbol{t}}
\newcommand{\bzero}{\boldsymbol{0}}
\newcommand{\util}{\boldsymbol{\tilde{u}}}
\newcommand{\utils}{\boldsymbol{\tilde{u}_*}}
\newcommand{\btil}{\boldsymbol{\tilde{\beta}}}
\newcommand{\btils}{\boldsymbol{\tilde{\beta}^{*}}}
\newcommand{\btilf}{(X'V^{-1}X)^{-1}X'V^{-1}\that}
\newcommand{\btilfs}{(X'V_*^{-1}X)^{-1}X'V_*^{-1}\that}
\newcommand{\bxij}{\boldsymbol{x_{ij}}}
\newcommand{\bxj}{\boldsymbol{x_{j}}}
\newcommand{\bei}{\boldsymbol{e_{i}}}
\newcommand{\bej}{\boldsymbol{e_{j}}}
\newcommand{\bek}{\boldsymbol{e_{k}}}
\newcommand{\bbary}{\boldsymbol{\bar{y}}}
\newcommand{\thet}{\boldsymbol{\theta}}

\newcommand{\vv}        {V^{-1}}
\newcommand{\vs}        {V^{-1}_*}
\newcommand{\vvs}        {V_{*}}
\newcommand{\sig}        {\Sigma}
\newcommand{\sm}        {\sqrt{m}}
\newcommand{\thhk}        {\mbox{$\hat{\theta}_k$}}
\newcommand{\thho}        {\mbox{$\hat{\theta}_1$}}
\newcommand{\thhm}        {\mbox{$\hat{\theta}_m$}}
\newcommand{\thj}        {\mbox{$\hat{\theta}_j$}}
\newcommand{\ttil}        {\mbox{$\tilde{\boldsymbol{\theta}}$}}

\newcommand{\thk}        {\mbox{$\hat{\theta}_k$}}
\newcommand{\thbi}        {\mbox{$\hat{\theta}_i^B$}}
\newcommand{\thbis}        {\mbox{$\hat{\theta}_{i*}^B$}}
\newcommand{\thebis}        {\mbox{$\hat{\theta}_{i*}^{EB}$}}
\newcommand{\thebjs}        {\mbox{$\hat{\theta}_{j*}^{EB}$}}
\newcommand{\thbj}        {\mbox{$\hat{\theta}_j^B$}}
\newcommand{\thbjs}        {\mbox{$\hat{\theta}_{j*}^B$}}
\newcommand{\thbk}        {\mbox{$\hat{\theta}_k^B$}}
\newcommand{\thebi}       {\mbox{$\hat{\theta}_i^{EB}$}}
\newcommand{\thebj}       {\mbox{$\hat{\theta}_j^{EB}$}}
\newcommand{\theblupi}    {\mbox{$\hat{\theta}_i^{EBM1}$}}
\newcommand{\thbmi}    {\mbox{$\hat{\theta}_i^{BM1}$}}
\newcommand{\theblupis}       {\mbox{$\hat{\theta}_{i*}^{EBM1}$}}

\newcommand{\thhw}       {\mbox{$\bar{\hat{\theta}}_{w}$}}
\newcommand{\thbarw}     {\mbox{$\bar{\hat{\theta}}_{w}^{B}$}}
\newcommand{\thbarwb}     {\mbox{$\bar{\hat{\theta}}_w^{B}$}}
\newcommand{\thbarwbs}     {\mbox{$\bar{\hat{\theta}}_{w*}^{B}$}}
\newcommand{\thbarweb}     {\mbox{$\bar{\hat{\theta}}_w^{EB}$}}
\newcommand{\thbarwebs}     {\mbox{$\bar{\hat{\theta}}_{w*}^{EB}$}}
\newcommand{\sbb}     {\mbox{$\sigma_b^2$}}
\newcommand{\sut}     {\mbox{$\tilde{\sigma}_u^2$}}
\newcommand{\suh}     {\mbox{$\hat{\sigma}_u^2$}}
\newcommand{\sus}     {\mbox{${\sigma}_u^{*2}$}}
\newcommand{\sust}     {\mbox{$\tilde{\sigma}_u^{*2}$}}
\newcommand{\hvis}     {\mbox{$\boldsymbol{x}_i'(X'V^{-1}_*X)^{-1}\boldsymbol{x}_i$}}
\newcommand{\hvi}     {\mbox{$\boldsymbol{x}_i'(X'V^{-1}X)^{-1}\boldsymbol{x}_i$}}
\newcommand{\hij}     {\mbox{$\boldsymbol{x}_i'(X'X)^{-1}\boldsymbol{x}_j$}}
\newcommand{\hvk}     {\mbox{$\boldsymbol{x}_k'(X'V^{-1}X)^{-1}\boldsymbol{x}_k$}}
\newcommand{\hj}     {\max_{1\leq j \leq m} h_j}
\newcommand{\hi}     {\max_{1\leq i \leq m} h_i}
\newcommand{\hk}     {\mbox{$\boldsymbol{x}_k'(X'X)^{-1}\boldsymbol{x}_k$}}
\newcommand{\hvik}     {\mbox{$\boldsymbol{x}_i'(X'V^{-1}X)^{-1}\boldsymbol{x}_k$}}
\newcommand{\hii}     {\mbox{$\boldsymbol{x}_i'(X'X)^{-1}\boldsymbol{x}_i$}}
\newcommand{\ut}     {\mbox{$\tilde{\boldsymbol{u}}$}}
\newcommand{\uts}     {\mbox{$\tilde{\boldsymbol{u}}_*$}}
\newcommand{\ub}     {\mbox{${\boldsymbol{u}}$}}
\newcommand{\li}     {\mbox{${\lambda_i}$}}
\newcommand{\lj}     {\mbox{${\lambda_j}$}}
\newcommand{\lk}     {\mbox{${\lambda_k}$}} 
\newcommand{\co}     {\text{Cov}}
\newcommand{\lp}     {\left(}
\newcommand{\rp}     {\right)}
\newcommand{\lb}     {\left[}
\newcommand{\rb}     {\right]}
\newcommand{\gos}     {G_1^{*}}
\newcommand{\gtos}     {G_2^{*}}
\newcommand{\gts}     {G_3^{*}}
\newcommand{\g}     {\mbox{$X(X'V^{-1}X)^{-1}X'$}}
\newcommand{\xg}     {\mbox{$(X'V^{-1}X)^{-1}X'$}}
\newcommand{\bw}{\boldsymbol{w}}
\newcommand{\bci}{\boldsymbol{c}_i}
\newcommand{\bgi}{\boldsymbol{g}_i}
\newcommand{\bxk}{\boldsymbol{x}_k}
\newcommand{\byi}{\boldsymbol{y}_i}
\newcommand{\bzi}{\boldsymbol{z}_i}
\newcommand{\bt}{\boldsymbol{\tilde{\beta}}}
\newcommand{\lai}     {\lambda_i}
\newcommand{\gai}     {\gamma_i}
\newcommand{\ma}     {\max_{1 \leq i \leq m}}
\newcommand{\mak}     {\max_{1 \leq k \leq m}}

\newcommand{\cov}     {\text{Cov}}


\newcommand{\zv}{{\bf z}_v}
\newcommand{\ua}{{\bf u}_a}
\newcommand{\uav}{{\bf u}_{A(v)}}
\newcommand{\sa}{\alpha^2_a}
\newcommand{\bet}{\boldsymbol\beta_v}
\newcommand{\sv}{\nu^2_v}
\newcommand{\se}{\sigma^2_v}
\newcommand{\R}{\mathbb{R}}
\newcommand{\bmu}{\boldsymbol\mu}
\newcommand{\bSigma}{\boldsymbol\Sigma}

\maketitle
\begin{abstract}
  We develop constrained Bayesian estimation methods for small area problems:
  those requiring smoothness with respect to similarity across areas, such as
  geographic proximity or clustering by covariates; and benchmarking
  constraints, requiring (weighted) means  of estimates to agree
  across levels of aggregation. We develop methods for constrained estimation
  decision-theoretically and discuss their geometric interpretation.  Our constrained estimators are the solutions to tractable optimization problems and
   have closed-form solutions.  Mean squared errors of the constrained
  estimators are calculated via bootstrapping. Our techniques are free of distributional assumptions and apply whether the  estimator is linear or non-linear, univariate or multivariate.  We illustrate our methods using data from the U.S. Census's Small Area Income and Poverty Estimates program.
\end{abstract}

\section{Introduction}

Small area estimation (SAE) deals with
estimating many parameters, each associated with an ``area''---a geographic
domain, a demographic group, an experimental condition, etc.  
Areas  are ``small'' since there is little or no information about any one area.
Estimates of a parameter based only on observations from the
associated area, called direct estimates, are imprecise.  To increase
precision, one tries to ``borrow strength'' from related areas, and
hierarchical and empirical Bayesian models are one way to do
so.
Since the pioneering work of
\cite{fay_1979, battese_1988}, such models have dominated 
SAE, with many
successful applications in official statistics,
sociology, epidemiology, political science, business, etc.~\citep{rao_2003}. Recently, 
SAE has been applied in other fields, such as neuroscience, where it has been shown to do as well as common approaches such as smoothed ridge regression and elastic net \citep{Wehbe-et-al-regularized-brain-reading}.

We extend 
these
classical
approaches in two directions, both of which have been the subject of recent
interest in the SAE literature.  One
direction
is to directly take account of information
about the proximity of areas in space or time.
In many applications, it is reasonable to expect that the parameters will be
smooth, so
that 
nearby areas
will
have similar parameters,
but
this is not
altogether standard within SAE \citep{rao_2003}.  Incorporating
spatial or temporal dependence directly into Bayesian models leads to
statistical and computational difficulties, yet it seems misguided 
to discard
such information.
The other direction is
``benchmarking,'' the imposition of consistency constraints on (weighted)
averages
of the parameter estimates.  A simple form of 
benchmarking is when the average of the parameter
estimates must match a known global average. 
When
there are multiple levels of aggregation for the
estimates, there can be issues of internal consistency as well.

We provide a unified approach to smoothing and benchmarking by regarding them
both as {\em constraints} on Bayes estimates.  Benchmarking corresponds to
equality constraints on global averages and variances.  Similarly, smoothing
corresponds to an inequality constraint on the ``roughness'' of estimates (how
much the parameter estimates of nearby areas differ).  The motivation of this smoothing is based upon manifold learning and frequentist non-parametrics, where loss functions are augmented by a penalty. Such a penalty term is in the spirit of ridge regression, where a transformation of the parameters is performed and additional shrinkage is carried out.  Our penalty corresponds to 
 how much estimates at nearby points in the domain
should tend to differ.

Decision-theoretically,
we obtain smoothed, benchmarked estimates by minimizing the Bayes risk subject
to these constraints, extending the approaches of \cite{datta_2011,ghosh_2013}
(themselves in the spirit of \cite{louis_1984} and \cite{ghosh_1992}).
Geometrically, the constrained Bayes estimates are found by projecting the
unconstrained estimates into the feasible set.  If the constraints are linear, then the the resulting optimization can be solved in closed form, requiring nothing more than basic matrix
operations on the unconstrained Bayes estimates.
 When we use
equality
constraints that are quadratic in nature, the problem cannot be solved in closed form, and the optimization is in fact non-convex.

Previous efforts at smoothing in SAE
problems have smoothed either the raw data or direct estimates.  In contrast,
we smooth estimates based on models which do not themselves include spatial
structure.  Computationally, this is much easier than expanding the models. Our optimization problems
can be solved in closed form and retain the advantages of
model-based estimation.
This approach to smoothing also combines naturally with the imposition of benchmarking constraints.

Another strong advantage of our decision-theoretic and geometric approach is its
extreme generality.  We require no distributional assumptions on the data or on
the unconstrained Bayes estimator.  Our results apply whether the unconstrained
estimator is linear or non-linear, whether the parameters being estimated are
univariate or multivariate, and whether there is a single level of aggregation
(``area''-level models) or multiple (``unit''-level models).  The relevant
notion of proximity between areas may be spatial, temporal, or more abstract.
It can even include clustering on covariates not directly included in the
Bayesian model.

\S
\ref{sec:smooth} describes our approach to smoothing in small area estimation.
This is extended to add benchmarking constraints in
\S \ref{sec:benchmarking}.  In both sections, area- and unit-level
results are derived using a single framework.
\S\ref{sec:bootstrap} discusses uncertainty quantification based on a bootstrap approach.
\S \ref{app:saipe} presents
an application to estimating the number of children living in poverty in each
state, using data from the U.S.\ Census Bureau's Small Area Income and Poverty
Estimates Program (SAIPE).  

\subsection{Notation and Terminology}
\label{sec:notation}

We assume $m$ \textbf{areas}, and for each area $i$, we estimate an associated
scalar quantity $\theta_i$, collectively $\bm{\theta}$.  ``Areas'' are often
spatial regions, where they might be different
demographic groups or experimental conditions.  Allowing $\theta_i$ to be
vectors rather than scalars is straightforward (see the remark at the end of \S
\ref{sec:smooth}).  Each area has a vector of
covariates $\bm x_i$, which may include spatial or temporal coordinates, when
applicable.  
Conditioning
a Bayesian model on an observed response $\bm{y}$ and covariates $\bm{x}$ leads to a Bayes estimate $\thbi$
for each area.  (Note that $\thbi$ is obtained by conditioning on {\em all} the
observations and covariates, not just those of area~$i$.)

The loss function is weighted squared error, where the weight for area~$i$ is
$\phi_i > 0 $, and the total loss from the action (estimate) $\bdd$ is
$\sum_i{\phi_i (\theta_i - \delta_i)^2}.$
In many
SAE applications, $\phi_i$ reflect variations in measurement precision and can
be obtained from the survey design \citep{rao_2003,pfeffermann_2013}. We assume
they are known (however, in practice they usually must be estimated). Define $\Phi
= \text{Diag}(\phi_i)$, which is positive definite by construction.

\subsection{Related Work}

\cite{pfeffermann_2013} provided a comprehensive review of the SAE and
benchmarking literature. Our work is twofold: smoothing SAE, and its combination with benchmarking.  Our SAE approach is decision-theoretic  with the addition of a
\emph{smoothness penalty} in the loss function. Our approach to benchmarking
with smoothing generalizes the benchmarking work of \cite{datta_2011,ghosh_2013}.

It is thought that spatial correlations {\em may} help SAE
models, leading to approaches such as correlated sampling error, spatial
dependence of small area effects, time series components, etc., reviewed in
\cite{ghosh_1994}.  More recently, spatially-correlated random effects have
been incorporated into empirical Bayesian models \citep{pratesi_2008} and into
hierarchical Bayesian models \citep{souza_2009}.  These approaches have all
been highly application-specific and hard to integrate with benchmarking, and
they greatly increase the computational cost of obtaining estimates.  Our goal
is to overcome these limits by taking a radically different approach.

Thus, we employ ideas about smoothing on graphs and manifolds from
frequentist non-parametrics and machine learning.  In particular, we take
advantage of ``Laplacian'' regularization ideas
\citep{Belkin-Niyogi-Sindhwani,Corona-et-al-laplacian-rkhs,Lee-Wasserman-spectral-connectivity},
where the loss function is augmented by a penalty term which reflects how much
estimates at nearby points in the domain should tend to differ.  Such regularization is
designed to ensure that estimates vary smoothly with respect to the intrinsic
geometry of some underlying graph or manifold. (Smoothness on a domain is
represented mathematically by the domain's Laplacian operator, which is the
generator for diffusion processes.)  This generalizes the roughness or
curvature penalties from spline smoothing \citep{Wahba-spline-models}
to domains more geometrically complicated than~$\mathbb{R}^m$.  We are unaware
of any previous application of Laplacian regularization to SAE problems,
though spline smoothing is often used in spatial statistics, including such
classic SAE tasks as estimating disease rates
\citep{Kafadar-smoothing-geographical}.


\section{Smoothing for Small Area Estimation}
\label{sec:smooth}

In this section, we develop
estimators that minimize
posterior risk while still
imposing smoothness on the estimate. The kind of smoothing we impose derives
from the literature on Laplacian regularization and semi-supervised learning
\citep{Belkin-Niyogi-Sindhwani,Corona-et-al-laplacian-rkhs,Lee-Wasserman-spectral-connectivity}.
The estimators we derive do not depend on the distributional assumptions of the
Bayesian models and are equally applicable to spatial smoothing or more
abstract clustering.  We do assume that the smoothing or clustering is done
separately from the estimation for each area or domain, and we also take weighted
squared error as the loss function.  In
\S \ref{sec:benchmarking}, we extend our approach to include
benchmarking.

\subsection{General Result}


We begin by introducing the symmetric matrix $Q$, with elements $q_{i\iprime} \geq 0$, to gauge how important it is that the estimate of $\theta_i$ be close to the estimate of $\theta_{\iprime}$.
It may often be the case that $q_{i\iprime}=q(\bm x_i, \bm x_{\iprime})$, i.e., the degree of smoothing of $\delta_i$ and $\delta_{\iprime}$ is a function of the covariates $\bm x_i$ and $\bm x_{\iprime}$.
Note also that the $q_{i\iprime}$ may be discrete-valued, corresponding to
clustering of areas, or continuous-valued, corresponding to a metric space of
areas.

A natural measure of the smoothness of $\bdd$ is the
$Q$-weighted sum of squared differences between elements, $
\sum_{i,\iprime}{(\delta_{i}-\delta_{\iprime})^2 q_{i\iprime}}.$
Hence, we add
a penalty
term $\gamma \sum_{i,\iprime}{(\delta_{i}-\delta_{\iprime})^2 q_{i\iprime}}$ to
our objective function, with the penalty factor~$\gamma\ge0$ chosen to
specify the overall importance of smoothness.
(We address the choice of~$Q$ below and of~$\gamma$ in \S \ref{sec:xv}.)


Therefore, we seek to minimize the posterior risk of the loss function
\begin{equation}
  \label{sa_sp}
  L(\bm{\theta}, \bdd) = \sum_i{\phi_i (\theta_i - \delta_i)^2} + \gamma
  \sum_{i,\iprime}{(\delta_i - \delta_\iprime)^2 q_{i\iprime}}.
\end{equation}
Minimizing the posterior expectation of \eqref{sa_sp} is equivalent to
minimizing
\begin{equation}
  \label{fh_sp}
   \sum_i \phi_i E[(\theta_i - \delta_i)^2 | \bmy ] +
    \gamma \sum_{i,\iprime} (\delta_i - \delta_\iprime)^2 q_{i\iprime}.
\end{equation}
Define $\Omega $ to be a matrix such that
$ \sum_{i,\iprime}{(\delta_i - \delta_\iprime)^2 q_{i\iprime}} = \bm{\delta}^T \Omega \bm{\delta} .$ (See Lemma~\ref{lemma:squared-differences} for details.)
Then minimizing~(\ref{fh_sp}) is equivalent to minimizing
\begin{equation}
  \label{fh_sp_2}
  (\bdd - \hat{\bm{\theta}}^B)^T\Phi (\bdd - \hat{\bm{\theta}}^B) + \gamma\bdd^T \Omega
  \bdd.
\end{equation}
 (We refer to \citep{datta_2011, ghosh_2013} for details on this equivalence.) 
Then we have the following result.

\begin{theorem}
  \label{area}
  \label{thm:general-smoothing}
  The smoothed Bayes estimator is
  \[
  \tilde{\bm\theta}^S = (I_m + \gamma\Phi^{-1} \Omega)^{-1}\hat{\bm{\theta}}^B.
  \]
\end{theorem}
\begin{proof}
  Differentiating \eqref{fh_sp_2} and setting the
  gradient to zero at $\tilde{\bm\theta}^S$ yields
$\Phi(\tilde{\bm\theta}^S - \hat{\bm{\theta}}^B) + \gamma \Omega \tilde{\bm\theta}^S = \bm0$. Then
\par\vspace*{-\parskip}\noindent
  \begin{align*}
 (\Phi + \gamma \Omega) \tilde{\bm\theta}^S = \Phi \hat{\bm{\theta}}^B
  &\implies \tilde{\bm\theta}^S
=(I_m + \gamma\Phi^{-1} \Omega)^{-1}\hat{\bm{\theta}}^B.
 \end{align*}
 Since
 \eqref{fh_sp_2} is a positive-definite quadratic form in $\bdd$, the solution is
 unique.
\end{proof}

See~\S\ref{sec:unit-level-smoothing} for an extension to unit-level models.

\begin{remark}
The parameter to be estimated for each area may be multivariate.  For instance,
we might seek both a poverty rate and a median income for each area.  For
simplicity, we assume that the parameter dimension $p$ is the same for each of
the $m$ areas.  Then Theorem \ref{thm:general-smoothing} can be applied with
$\bm{\theta} = (\theta_{11}, \ldots,
\theta_{m1},\ldots,\theta_{12},\ldots,\theta_{m2}, \ldots, \theta_{1p},\ldots, \theta_{mp})$.
The matrix~$\Phi$ remains the diagonal matrix of the $\phi_{ij}$, in the same order  as $\bm{\theta}$.  However, $\Omega$ is now a
block-diagonal matrix, where each $m\times m$ block contains a copy of the appropriate matrix for the corresponding univariate problem.  This ensures that the
same smoothness constraint is imposed on each component of the parameter
vectors, but different components are not smoothed together.
\end{remark}

\section{Benchmarking and Smoothing}
\label{sec:benchmarking}

We now turn to situations where our estimates should not just be smooth,
minimizing \eqref{fh_sp_2}, but also obey benchmarking
constraints.  As the benchmarking constraints are relaxed, we should recover
the results of \S \ref{sec:smooth}.  Our approach to finding benchmarked
Bayes estimators extends that of \cite{datta_2011,ghosh_2013}.  We employ the
following definition.

\begin{definition}[Benchmarking constraints, benchmarked Bayes estimator]
  \emph{Benchmarking constraints} are equality constraints on the weighted
  means or weighted variances of subsets (possibly all) of the estimates.  The
  \emph{benchmarked Bayes estimator} is the minimizer of the posterior risk
  subject to the benchmarking constraints.
\end{definition}

The levels to which we benchmark, i.e., the values of the equality constraints, are assumed to 
be given \emph{externally} from some other data
source.
\textcolor{black}{(For internal benchmarking, see \cite{bell_2013}.)
Our methods address linear, weighted mean constraints, as in
\cite{datta_2011, ghosh_2013}.}


\subsection{Linear Benchmarking Constraints}

We first consider benchmarking constraints which are linear in the estimate~$\bdd$, such as means or totals.  The general problem is now to minimize the posterior risk in~(\ref{fh_sp_2}) subject to the constraints
\begin{equation}
M\bdd=\bm t,\label{eqn:linear-constraints}
\end{equation}
where $\bm{t}$ is a given $k$-dimensional vector and $M$ is a $k\times m$
matrix.
As before, this is equivalent to
introducing a $k$-dimensional vector of Lagrange multipliers $\bm{\lambda}$ and minimizing
\[
(\bdd-\hat{\bm{\theta}}^B)^T \Phi (\bdd-\hat{\bm{\theta}}^B) + \gamma \bdd^T \Omega \bdd -2\bm{\lambda}^T(M\bdd - \bm{t}).
\]

\begin{theorem}
  Suppose
  (\ref{eqn:linear-constraints})
  has solutions.  Then the constrained Bayes estimator
  under~(\ref{eqn:linear-constraints}) is
  \small{
  \[
  \tilde{\bm\theta}^{BM} \!= \Sigma^{-1}\!\left[ \Phi \hat{\bm{\theta}}^B +
    M^T (M \Sigma^{-1} M^T)^{-1} \!\left ( \bm{t} - M \Sigma^{-1} \Phi
      \hat{\bm{\theta}}^B \right ) \right]\!,
  \]
  }
  where $\Sigma = \Phi + \gamma \Omega$.
  \label{thm:general-linear-benchmarked}
\end{theorem}


\begin{remark}
Note that the Theorem~\ref{thm:general-linear-benchmarked} estimator~$\tilde{\bm\theta}^{BM}$ can be expressed in terms of the Theorem~\ref{thm:general-smoothing} estimator~$\tilde{\bm\theta}^S$ as
\[
\tilde{\bm\theta}^{BM} = \tilde{\bm\theta}^S+ \Sigma^{-1} M^T 
(M \Sigma^{-1} M^T)^{-1} (\bm t - M\tilde{\bm\theta}^S).
\]
Thus, it can be seen that the benchmarking essentially ``adjusts'' the estimator $\tilde{\bm\theta}^S$ based on the discrepancy between $M\tilde{\bm\theta}^S$ and the target~$\bm t$.
\end{remark}


\begin{proof}[Proof of Theorem~\ref{thm:general-linear-benchmarked}]
  Differentiating with respect to $\bdd$ and setting the result equal to zero at $\tilde{\bm\theta}^{BM}$ yields
 \begin{align*}
  M^T \bm{\lambda} &= \Phi
  (\tilde{\bm\theta}^{BM}-\hat{\bm{\theta}}^B) + \gamma \Omega \tilde{\bm\theta}^{BM} \\
  &\implies \tilde{\bm\theta}^{BM} = \Sigma^{-1} (\Phi \hat{\bm{\theta}}^B +
  M^T \bm{\lambda} ).
\end{align*}
  Then by the constraint,
  \begin{align}
  &\bm t=M\Sigma^{-1}(\Phi \hat{\bm{\theta}}^B+M^T\bm{\lambda})\\
  &=M\Sigma^{-1}\Phi\hat{\bm\theta}^B+M\Sigma^{-1}M^T\bm\lambda,\notag
  \end{align}
  so $\bm\lambda=(M\Sigma^{-1}M^T)(\bm t-M\Sigma^{-1}\Phi\hat{\bm\theta}^B)$.
  The result follows immediately.
\end{proof}
Often there is only one linear constraint of the form $\sum_i w_i\delta_i=t$, or equivalently $\bm w^T\bm\delta=t$, for some nonnegative weights $w_i$ and some $t\in\mathbb{R}$.  This is simply a special case of Theorem~\ref{thm:general-linear-benchmarked} with $k=1$ and $M=\bm w^T,$ in which case the result simplifies to $$\tilde{\bm\theta}^{BM} = \tilde{\bm\theta}^S + (t - \bm w^T\tilde{\bm\theta}^S)(\bm w^T \Sigma^{-1} \bm w)^{-1} \Sigma^{-1} \bm w. $$

Also see~\S\ref{sec:multiple-weighted-means} for an extension to unit-level models.

\subsection{Weighted Variability Constraints}

Now suppose that we impose one additional constraint on the weighted variability of the
form
\begin{equation}
(\bdd - \bm{\tau})^T W (\bdd - \bm{\tau}) = h 
\label{eqn:quadratic-constraints}
\end{equation}
where $\bm\tau$ is an $m$-dimensional vector,
$W$ is a symmetric $m\times m$ matrix, and $h$ is a non-negative constant.
In the important special case of a single weighted-mean constraint of the form $\bm w^T\bm\delta=t$, the vector~$\bm\tau$ and the matrix~$W$ in~(\ref{eqn:quadratic-constraints}) are typically taken as $\bm\tau=t\bm1_m$ and $W=\text{Diag}(w_i)$, where $\bm1_m$ denotes the ones vector of length~$m$.

The imposition of such a constraint immediately renders the associated optimization problem considerably more challenging.
Specifically, closed-form solutions for the minimizer can no longer be obtained in general.
Moreover, notice that the set of $\bm\delta\in\mathbb R^m$ that satisfy both (\ref{eqn:linear-constraints}) and~(\ref{eqn:quadratic-constraints}) is no longer convex.
Thus, even a numerical solution may be difficult to obtain.  The question of how to incorporate variability constraints while maintaining tractability of the model is a potential direction of future research and is beyond the scope of this paper.

\subsubsection{Geometric Interpretation}
\label{geometric}

Our formulation of benchmarking and smoothing as constrained optimization
problems has a very simple geometric interpretation.  It is well known that the
Bayes estimate is the minimizer of the conditional expectation of the mean
squared error (MSE). Since the minimization is taken over {\em all} possible
values of $\bm{\theta},$ the Bayes estimate will not respect any constraints we
might wish to impose (except by chance) or unless these constraints are included in
the specification of the prior.  We instead seek to minimize the MSE
within the feasible set of the constraints.  We find the point in the feasible
set which is as close (in the sense of expected weighted squared error) to the Bayes estimate as possible.  That is, we {\em
  project} the Bayes estimate into the feasible set.

The geometry of the feasible set is itself slightly complicated, because of the
constraints imposed.
Note that the smoothness penalty in the loss function may be reformulated as a smoothness constraint of the form $\bdd^T \Omega\bdd \leq s$ for some $s>0$.  This constraint
defines an ellipsoid centered at the origin.  Constraints on
weighted means define linear sub-spaces, e.g., planes, depending on
the number of constraints and the number of variables.  Finally, constraints of
weighted variabilities define the surfaces of cones.  The constrained Bayes
estimator is the projection of the unconstrained Bayes estimator onto the
intersection of the ellipsoid, the linear sub-space, and the cones.

\subsection{Choice of Smoothing Penalties}
\label{sec:xv}

 The choice of~$\gamma$\footnote{In unit-level problems, $\gamma_A$ and
  $\gamma_U$; we will not note all the small modifications needed to pick two
  smoothing factors at once.}  is assumed fixed {\em a priori}.  
But knowing $\gamma$ is equivalent to knowing how smooth the estimate \emph{ought}
to be, and this knowledge is lacking in most applications.  In such situations,
we suggest obtaining $\gamma$ by leave-one-out cross-validation
\citep{Stone1974,Wahba-spline-models,Corona-et-al-laplacian-rkhs}.

For each value of $\gamma$ and each area $i$, define $\bm\delta^{(-i)}(\gamma)$ as
the solution of the corresponding optimization problem with the loss-function
term for~$i$ dropped\footnote{Instead of the sum of squared errors
  $\sum_{i'=1}^{m}{\phi_{i'} (\delta_{i'}-\theta_{i'})^2}$, we use $\sum_{i'\neq i}{\phi_{i'}
    (\delta_{i'}-\theta_{i'})^2}$.  This amounts to replacing $\Phi$ with a matrix
  whose $i$th row and column are both 0.}.  The smoothness penalty
and any applicable benchmarking constraints are however calculated over the
{\em whole} of the vector $\bdd$, not just the non-$i$ entries.  (This ensures
that $\bm\delta^{(-i)}(\gamma)$ does meet all the constraints, while still making
a {\em prediction} about $\theta_i$.)   The cross-validation score of
$\gamma$ is
\[
V(\gamma) = \frac{1}{m}\sum_{i=1}^{m}\left[\delta^{(-i)}_i(\gamma) - \thatb_i\right]^2 \phi_i ,
\]
where $\delta^{(-i)}_i(\gamma)$ denotes the $i$th component of $\bm\delta^{(-i)}(\gamma)$,
and the minimizer of the cross-validation scores is
$
\hat{\gamma} = \argmin_{\gamma\ge0}{V(\gamma)}.
$

Direct evaluation of $V(\gamma)$ can be computationally costly. See
 \cite{Wahba-spline-models} for faster approximations, such
as ``generalized cross-validation.''

\section{Evaluation Using a Residual Bootstrap}
\label{sec:bootstrap}

It is traditional in small area estimation to report approximations to the
overall estimation error.  (One of the main motivations of using small area
methods is, after all, reducing the estimation error.)  This is generally a
challenging undertaking, since while methods like cross-validation can be used
to evaluate \emph{prediction} error in a way which is comparable across models,
they do not work for \emph{estimation} error.  Thus, one needs to use more
strictly model-based approaches, either analytic or based on the
bootstrap.

Evaluating the MSE of our estimates is especially difficult, since we combine
a model-based estimate with a non-parametric smoothing term.  A straightforward
model-based bootstrap would sample from the posterior distribution of
\eqref{eqn:tlnise} to generate a new set of true poverty rates $\bm\theta^*$ and
observations $\bm y^*$, re-run the estimation on $\bm y^*$, and see how close the
resulting estimates $\bm\delta^*$ came to $\bm\theta^*$.  However, this presumes the
correctness of the Fay-Harriot model in \eqref{eqn:tlnise}, which is precisely
what we have chosen \emph{not} to assume through our imposition of the benchmarking/smoothing constraints\footnote{If we follow this procedure nonetheless, we always conclude
  that benchmarking and especially smoothing radically increase the MSE by
  introducing large biases.}.  Note that such constraints do not
 fit naturally into the generative model.

We evade this dilemma by using a semi-parametric residual bootstrap, a common
approach when the functional form of a regression is known fairly
securely, but the distribution of fluctuations is not. We calculate
the differences between $y_i$ and our constrained Bayes estimates
$\tilde{\theta}^{BM}_i$, resample these residuals, scale them to account for
heteroskedasticity, and add them back to $\tilde{\theta}^{BM}_i$ to generate
$y^*_i$.  (See Appendix \ref{app:bootstrap}.)  The residual
bootstrap assumes that smoothing is appropriate and that we
have chosen the right $\Omega$ matrix.

\section{Application to the SAIPE Dataset}
\label{app:saipe}

We apply our constrained Bayes estimation procedure to data from the Small Area
Income and Poverty Estimates (SAIPE) program of the U.S.\ Census Bureau.  Our
goal is to estimate, for 1998, the rate of poverty of children aged 5--17
years in each state and the District of Columbia.  This is an area-level
model, with states as the areas. The small area model from which we derive our
initial Bayes estimates is described in \S \ref{sec:fay-harriot-for-saipe}.
The primary benchmarking constraint is that the weighted mean of the state
poverty-rate estimates must match the national poverty rate established by
direct estimates.  A secondary benchmarking constraint is the matching of the
similarly-known national variance in poverty rates.  This benchmarking had
already been considered, for this data and model, in \cite{datta_2011}.  We add
to this the constraint of smoothness across states, where our choice of
Laplacian and of smoothing penalty is discussed in \S
\ref{sec:smoothing-the-application}. 


SAIPE estimates average household income and poverty rates from
small areas in the U.S. over multiple years.  It works by combining direct
estimates of these quantities, from the Annual Social and Economic (ASEC)
Supplement of the Current Population Survey (CPS) and the American Community
Survey (ACS), with standard small area models, which use as predictors several
variables drawn from administrative records.  This presumes that areas with
similar values of the predictor variables should have similar values of the
parameters of interest.  See \cite{bell_2013, datta_2011} for a fuller account
of the SAIPE program.

For illustrative purposes, and following \cite{datta_2011}, we have focused on
estimating the rate of poverty among children aged 5--17 in 1998.  The
public-use data here is at the state level, so states are areas.  The
predictor variables used are: a pseudo-estimate of the child poverty
rate based on Internal Revenue Service (IRS) income-tax statements; the rate of
households not filing taxes with the IRS\footnote{In the U.S., households whose
  income falls below certain thresholds are not required to file federal
  taxes.}; the rate of food stamp\footnote{A program providing direct
  assistance in buying food and other necessities for low-income households.}
use; and the residual term from a regression of the 1990 Census estimated child
poverty rate.  The last variable is supposed to help represent whether a state
has an unusual level of poverty, given its other characteristics, which is
presumably persistent over time.

\subsection{Hierarchical Bayesian Model for SAIPE}
\label{sec:fay-harriot-for-saipe}

As in \cite{datta_2011}, our initial, unconstrained Bayes estimates for
poverty rates are derived from the following hierarchical model of \cite{fay_1979}:
\par\noindent
\begin{align}
  \label{eqn:tlnise}
  y_i \mid \theta_i &\sim \mathcal{N}(\theta_i, D_i); \; i=1,\ldots, m \\
  \theta_i \mid \bm\beta &\sim \mathcal{N}(\bm x_i^T \bm\beta, \sigma_u^2) \notag \\
  \pi(\sigma_u^2, \bm\beta) &\propto 1. \notag
\end{align}
Here $\theta_i$ is the true poverty rate for state $i$, $y_i$ is the direct
survey estimate, $D_i$ is the known sampling variance of the survey, $\bm x_i$ are
the predictors, $\bb$ is the vector of regression coefficients, and
$\sigma_u^2$ is the unknown modeling variance.  The posterior means and
variances of $(\bb,\theta_i,\sigma_u^2)$ are estimated by Gibbs sampling.

\subsection{Benchmarking and Smoothing Results}
\label{sec:smoothing-the-application}

\textcolor{black}{
We consider three different possibilities for benchmarking and/or smoothing: (i)~benchmarking the mean alone without smoothing, (ii)~benchmarking both the mean and variability without smoothing, and (iii)~benchmarking the mean alone with smoothing.
Note that since there is no smoothing in~(ii), solutions can indeed be found in closed form; see~\cite{datta_2011} for details.}

In each case, the benchmarking weights~$w_i$
are proportional to the estimated population of children aged 5--17 in each
state. Intuitive remarks regarding how to choose $h$ for benchmarking the weighted variability
are given in \cite{ghosh_2013, datta_2011}. We refer to these for technical details.
Figure~\ref{fig:benchmarking_man_vs_mean_and_variab} compares the unconstrained Bayes
estimator to the benchmarked Bayes estimators.  Poverty estimates change very
little when benchmarking the weighted mean alone, and only a little more when
we benchmark both mean and variability.

\begin{figure}
  \begin{center}
    \includegraphics[width=0.65\textwidth]{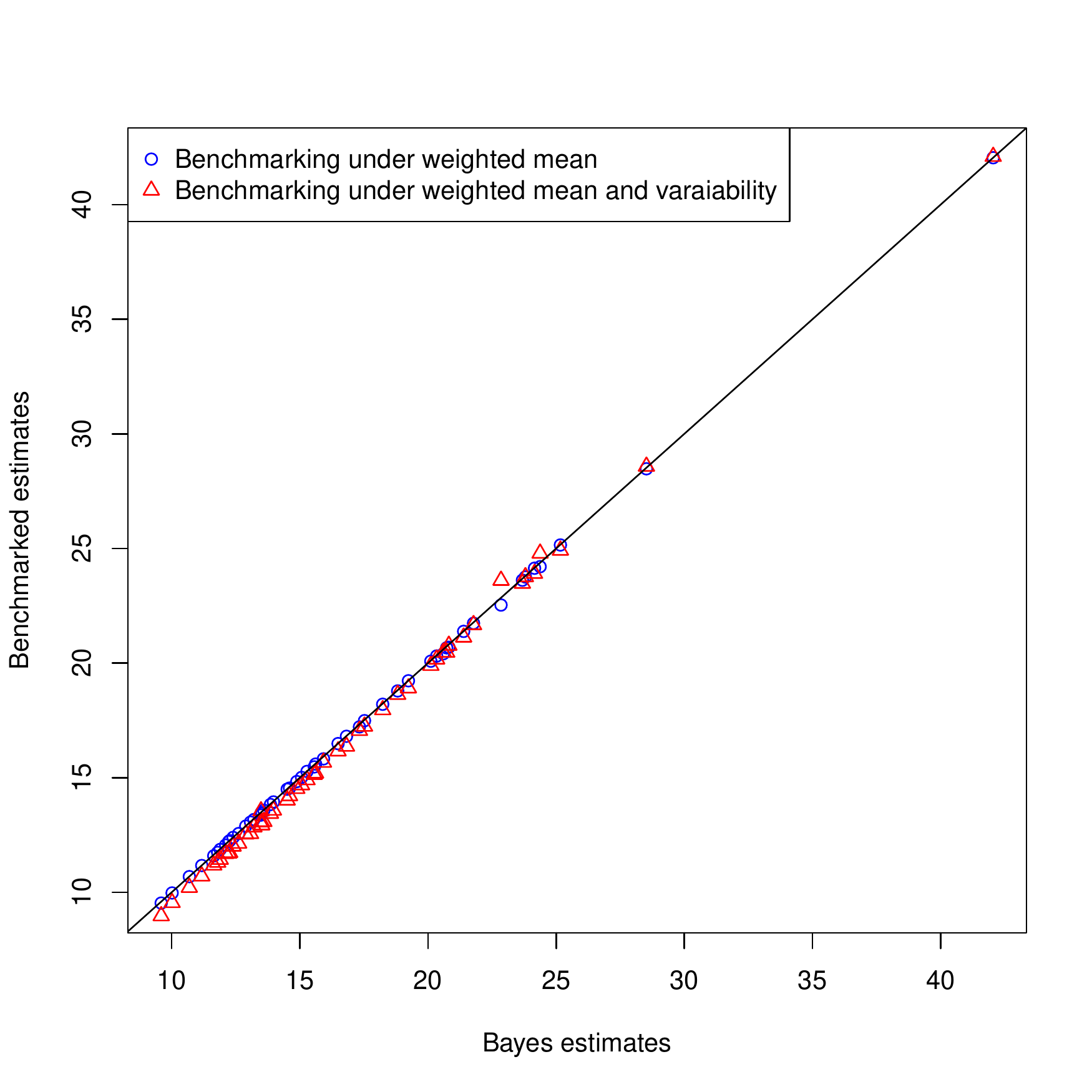}
    \caption{Benchmarking the mean alone leads to little change from the Bayes
      estimates; benchmarking both the mean and variability has very little improvement.
      }
    \label{fig:benchmarking_man_vs_mean_and_variab}
  \end{center}
\end{figure}


The most important part of our procedure is picking the matrix $\Omega$ used to
measure the smoothness of estimates---equivalently, picking the matrix $Q$
which says how similar the estimates for any two domains should be.  This is
inevitably application-specific.  In the results reported below, we used the
simple choice where $q_{i\iprime} = 1$ if the states $i$ and $\iprime$ shared a
border, and $0$ otherwise.  This treats the states as nodes in an unweighted
graph, with $Q$ being its adjacency matrix and $\Omega$ its Laplacian.  
As described
in \S \ref{sec:xv}, the smoothing factor $\gamma$ was picked by leave-one-out
cross-validation; the final value was $\gamma\approx 0.02$.  Figure~\ref{fig:adjustment} shows the smoothed and mean-benchmarked Bayes estimates
versus the unconstrained Bayes estimates.  In general, low Bayes estimates are
pulled up and high ones are brought down; everywhere, estimates are adjusted
towards their neighbors.

\begin{figure}
  \begin{center}
    \includegraphics[width=\textwidth]{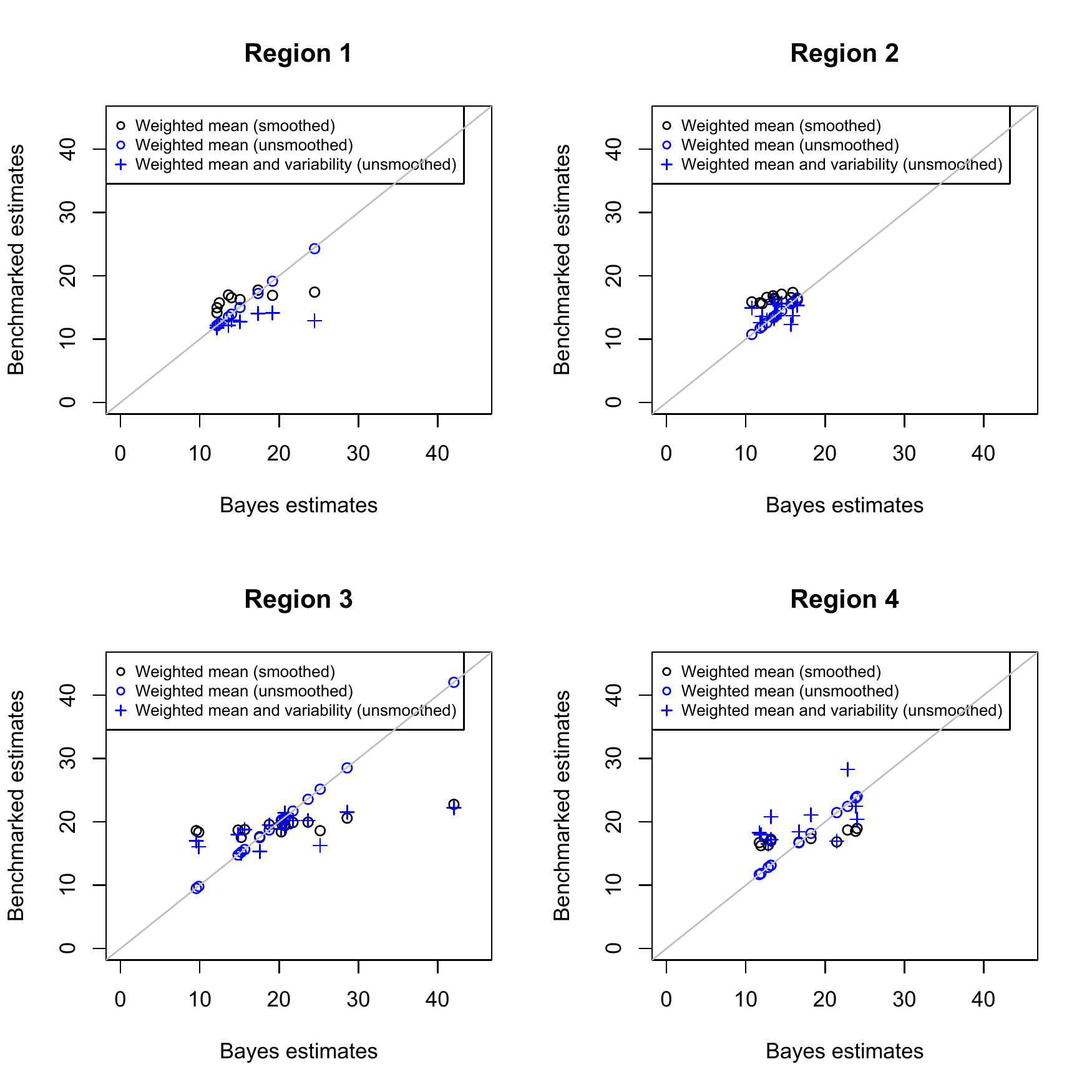}
    \caption{Benchmarked estimates with and without region, plotted against
the Bayes estimates, by region.}
    \label{fig:bench_smooth}
  \end{center}
\end{figure}

\begin{figure}
  \centering
  \includegraphics[width=0.8\textwidth]{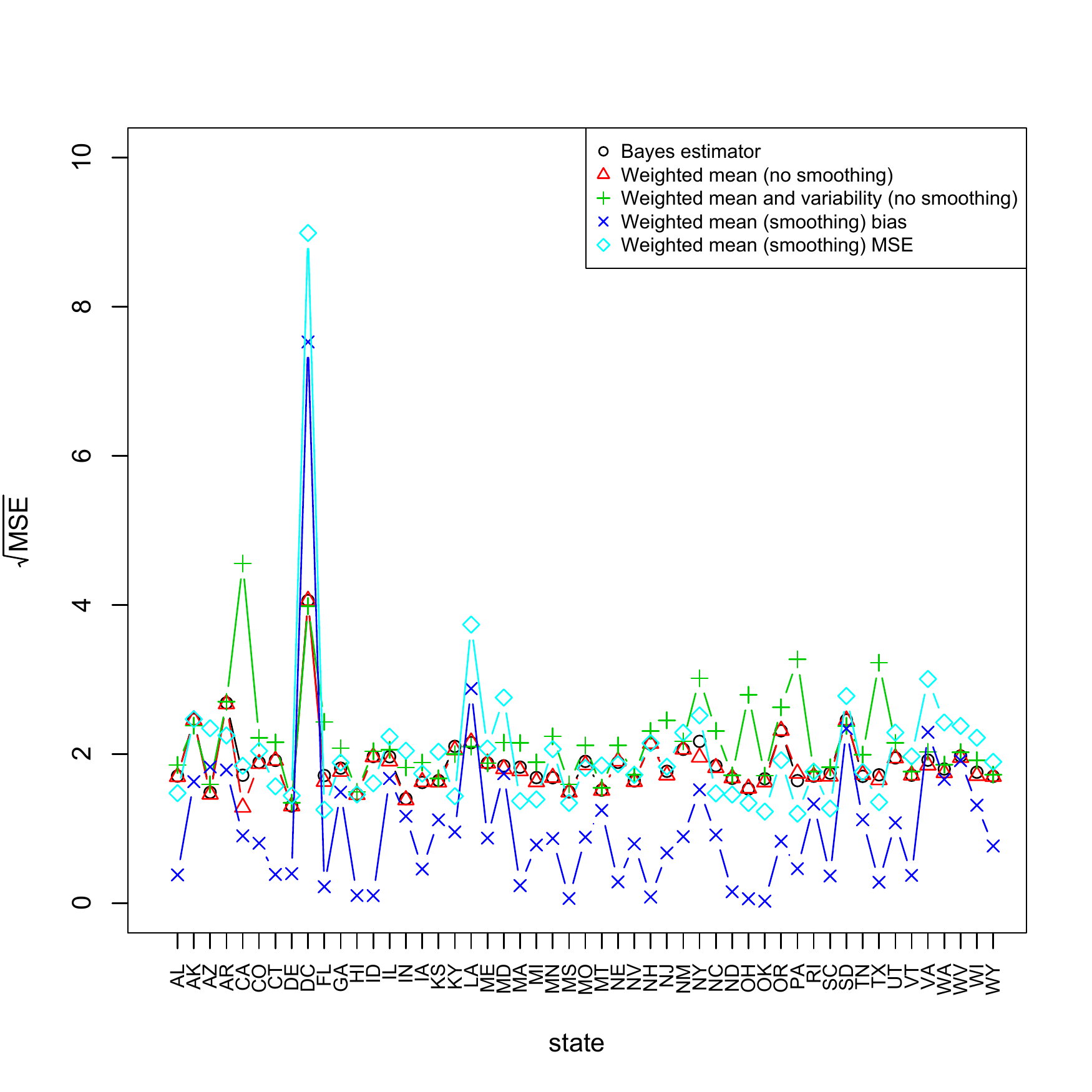}
  \includegraphics[width=0.8\textwidth]{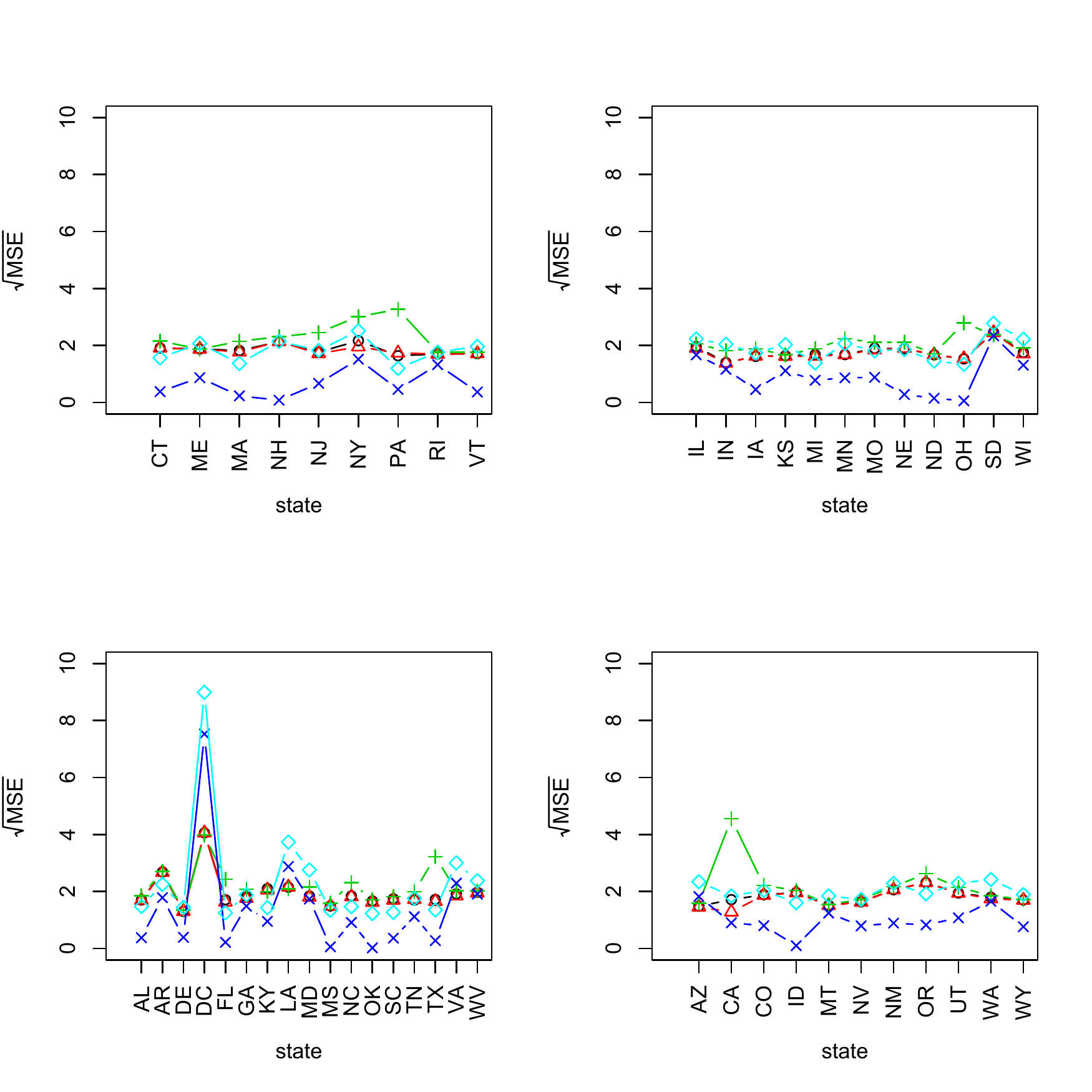}
  \caption{Above: Bootstrap MSEs for the SAIPE data and the Fay-Harriot model,
    under different combinations of benchmarking and smoothing.  Below: the
    same data, but broken into the geographic regions.}
  \label{fig:main-bootstrap-mses_new}
\end{figure}

\begin{figure}
  \begin{center}
    \includegraphics[width=0.6\textwidth]{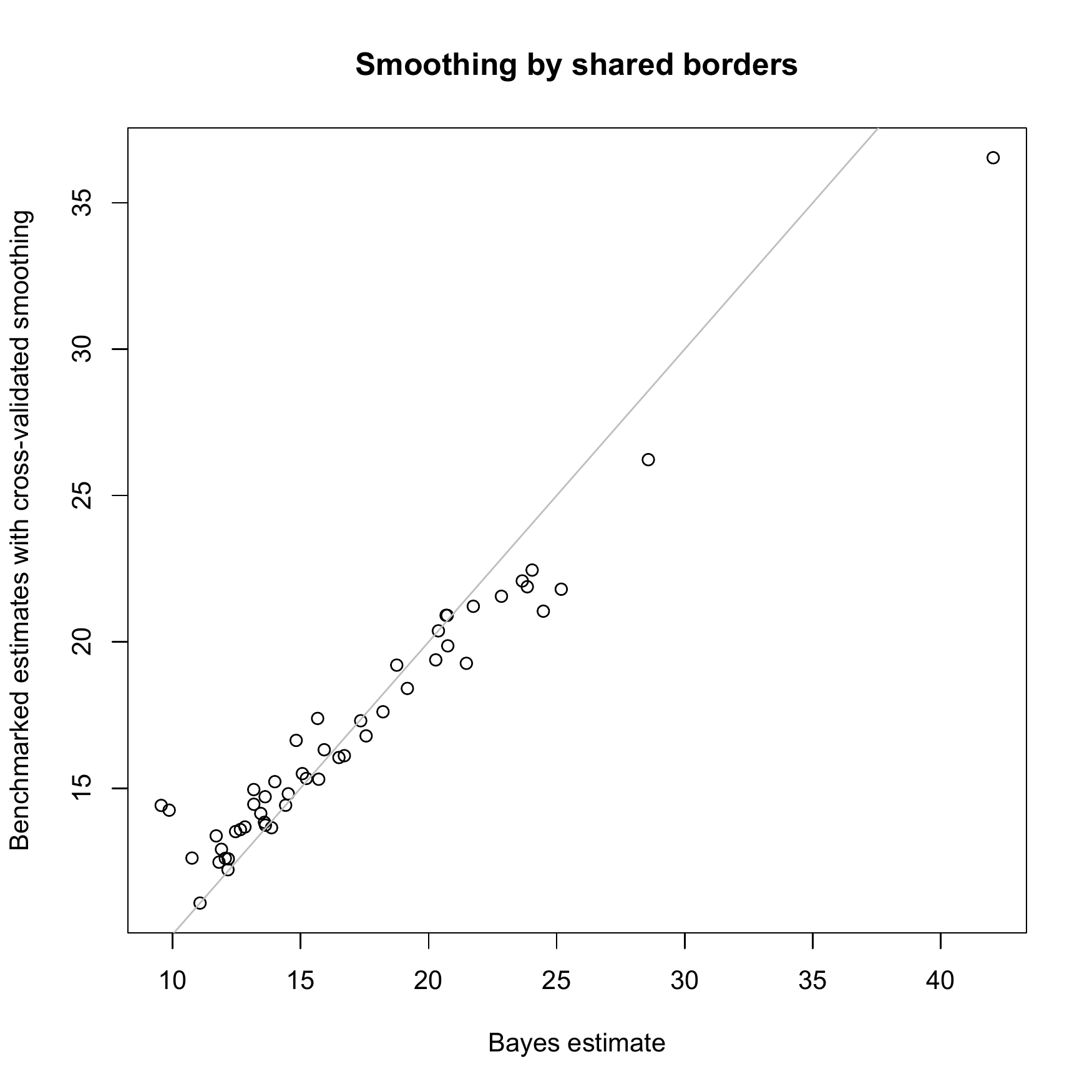}
    \caption{Smoothed, mean-constrained Bayes estimates versus unconstrained
      Bayes estimates.}
    \label{fig:adjustment}
  \end{center}
\end{figure}


Figure \ref{fig:bench_smooth} further illustrates the effects of combining
smoothing with benchmarking a weighted mean.  It is broken up by the four U.S.
Census regions\footnote{Region 1, the Northeast: Connecticut, Maine,
  Massachusetts, New Hampshire, New Jersey, New York, Pennsylvania, Rhode
  Island, Vermont.  Region 2, the Midwest: Illinois, Indiana, Iowa, Kansas,
  Michigan, Minnesota, Missouri, Nebraska, North Dakota, Ohio, South Dakota,
  Wisconsin.  Region 3, the South: Alabama, Arkansas, Delaware, the District of
  Columbia, Florida, Georgia, Kentucky, Louisiana, Maryland, Minnesota, North
  Carolina, Oklahoma, South Carolina, Tennessee, Texas, Virginia, West
  Virginia.  Region 4, the West: Arizona, California, Colorado, Idaho, Montana,
  Nevada, New Mexico, Oregon, Utah, Washington, Wyoming.  Note that Alaska and
  Hawaii are not included in this regional scheme.} for ease of visualization.
While benchmarking alone has relatively little impact on the Bayes estimates,
benchmarking plus smoothing does.  In each region, the smoothed estimates fall
on lines of slope less than 1, indicating shrinkage towards a common value,
even though the regions are not part of the smoothing scheme.  This means that
the value toward which the estimates are shrunk is not necessarily the regional mean---observe region~2, where most constrained estimates \emph{exceed} the Bayes
estimates.

Figure \ref{fig:main-bootstrap-mses_new} shows the statewise MSEs under the
bootstrap of \S\ref{sec:bootstrap} for different combinations of benchmarking and smoothing.  Smoothing
tends to bring down the bias and the MSE for most but not all states---it is
in fact known that the bias cannot be reduced uniformly across areas
\citep{ghosh_2013, pfeffermann_2013}.  The one ``state'' for which smoothing
drastically increases the estimated error is the District of Columbia, which is
unsurprising on substantive grounds.\footnote{Briefly, the District of Columbia is not a separate
  state, but rather a small part of a larger metropolitan area, containing a
  disproportionate share of the metropolis's poorest neighborhoods.  The
  adjoining states of Maryland and Virginia are much larger, much more
  prosperous, and much more heterogeneous.  Many of these issues would be
  alleviated if we had data on finer spatial scales.}


We considered several alternative ways of smoothing the Bayes estimates.  One
was to make $q_{i\iprime}$ decrease with the geographic distance betetween states, regarded
as points at either their centers or their capitals.  However, neither choice
of representative point was compelling, and we would also have to pick the
exact rate at which $q_{i\iprime}$ decreased with distance.  A second approach was to
treat the Census regions as clusters, setting $q_{i\iprime}=1$ within a cluster
and $q_{i\iprime}=0$ between them.  This however performed poorly under cross-validation.
Similarly, we attempted diffusion-map $k$-means\footnote{The covariates 
  were the state mean income, the state median income, the fraction of adults
  with at least a high-school education, the percentage of the population
  racially classified as white, and the percentage living in metropolitan
  areas.  
}
\citep{Lee-Wasserman-spectral-connectivity,Joey-Richards-diffusionMap} with
varying $k$, but none worked well under cross-validation.

We suspect that both distance-based and clustering approaches may work better
at finer levels of spatial resolution, e.g., moving from whole states to
counties or even census tracts, which are more demographically homogeneous.
However, this is speculative without such fine-grained data.


\section{Discussion}
We have provided a general approach to SAE at both the unit and area levels, where we
smooth and benchmark estimates. Our approach
\textcolor{black}{yields closed-form solutions without requiring}
any distributional assumptions. Furthermore, 
our results apply for linear and non-linear estimators and extend to multivariate settings.
Finally, we show through a bootstrap approximation and cross-validation that smoothing does improve estimation of poverty rates at the state level for the SAIPE dataset for most states as measured by MSE. 

\textcolor{black}{Note that we do not provide a simulation study, since any one that we considered was an unfair and biased comparison to either our proposed estimator or those proposed earlier in the literature. This is due to the fact that the Fay-Herriot model does not assume a smoothing/spatial component, however, our loss function does. We would be glad to consider any fair simulation study if one can be pointed out that would lead to helpful and meaningful results.}

\textcolor{black}{Another direction for future research is the extension of the present work to address weighted variability constraints as well.}
This becomes a difficult non-convex optimization problem,
\textcolor{black}{for which it is not clear how to efficiently and reliabiliy obtain a numerical solution}.
\textcolor{black}{The imposition of more than one weighted variability constraint specifies the feasible set as the intersection of multiple $(m-1)$-dimensional manifolds in $m$-dimensional Euclidean space.  Careful consideration of the geometry of the resulting optimization problem may yield insight into methods of obtaining exact or approximate solutions, at least in certain special cases.  Such ideas are clearly a potential direction for future work.}


Throughout, we have worked with squared error.  However, it should be possible
to replace this with any other convex norm, with minimal changes to our
approach.  Once the Bayes estimate is obtained, the constrained Bayes estimate
would be found by projection onto the feasible set.  This would presumably
mean more numerical optimization and fewer closed forms, but the optimization
would remain convex and tractable.  Getting the initial Bayes estimates under a
different loss function might be more challenging.

It may be possible to go beyond point estimates to distributional
estimates.  Given a sample from the posterior distribution (e.g., from MCMC),
it is possible to project each sample point into the feasible
set, giving a posterior distribution whose support respects the constraints.
The inferential validity of this sample would however require careful
investigation.\footnote{Note that this is rather different from the idea in the
  recent papers of \cite{zhu_2012} of regularizing the posterior
  distribution, where the constraints or penalties are expressed as functionals
  of the whole posterior distribution.}

\small{
\section*{Acknowledgements}  
  RCS was supported by NSF grants SES1130706 and DMS1043903 and NIH grant \#1 U24 GM110707-01.}
  
  \normalsize


  

\newpage

\bibliography{chomp,locusts}
\bibliographystyle{asa}

\clearpage
\newpage
\appendix

\section{Results for Unit-Level Models}

Many problems feature multiple levels of aggregation.
For simplicity, we consider the specific case of two levels (from which the extension to three or more levels will be fairly clear).
``Areas'' refer to the upper level of aggregation and are divided into \textbf{units}.  The $i$th area contains $n_i$ units; the total number of
units is $N = \sum_{i}{n_i}$.  Units are strictly nested within areas and are indexed by $j$.  We denote the area-level
quantities as $\theta^A_i$ (with covariates~$\bm x^A_i$, etc.), and the unit-level
parameters as $\theta^U_{ij}$ (with covariates~$\bm x^U_{ij}$, etc.).  Denote the
vectors of Bayes estimates by $\hat{\bm{\theta}}_A^B$ and
$\hat{\bm{\theta}}_U^B.$ The loss weight for unit $j$ in area $i$ is
$\xi_{ij}.$ Assume that loss is additive across areas and units; thus, the
total loss from the action (estimate) $(\bdd^A, \bdd^U)$ is
\[
\sum_i{\phi_i (\delta^A_i - \theta^A_i)^2} + \sum_{ij}{\xi_{ij} (\delta^U_{ij} - \theta_{ij}^U )^2}.
\]
Define $\Xi$ as the diagonal matrix of the $\xi_{ij}$, which is
positive-definite. 

In many important cases, the area-level parameters are functions (e.g.,
weighted means or proportions) of the parameters for the units contained within
the area (e.g., we might use $\overline{\theta}_{iw} =
\sum_{j}{w_{ij}\theta^U_{ij}}$ as our $\theta^A_i$).  Less trivial examples are
quantiles or Gini coefficients of the $\theta^U_{ij}$ for each area.  However,
it does not make sense for the unit-level parameters to be functions of the
area-level parameters.  The area-level parameter does not {\em have} to be a
function of the unit-level parameters (e.g., if we have random effects for both
areas and for units, the latter do not determine the former).

The general results of Theorems~\ref{area}~and~\ref{thm:general-linear-benchmarked}
can also be applied to models at the unit level, as described below.

\subsection{Smoothing for Unit-Level Models}
\label{sec:unit-level-smoothing}
\label{sec:partitioning}

Consider the case where each area is partitioned into units, and estimates
are sought at both the unit and the area level.  (See \S \ref{sec:notation} for
notation.)  We need two similarity functions, $q_A$ as before, and $q_U$, where
$q_U(x_{ij}, x_{\iprime\jprime})$ is the similarity between unit $j$ in
area~$i$ and unit $\jprime$ in area $\iprime$.  The smoothness-augmented loss
function is
\par\vspace*{-\parskip}\noindent
\begin{align}
&L(\bm{\theta}^A,\bm{\theta}^U, \bdd^A, \bdd^U) \notag \\
  &=  \sum_i{\phi_i (\delta_i^A - \theta^A_i)^2} + \sum_{ij}{\xi_{ij} (\delta_{ij}^U - \theta_{ij}^U )^2}  
 \notag \\
 & \quad+ \gamma_A \sum_{i, \iprime}{(\delta^A_{i} -\delta^A_{\iprime})^2 q^A_{i\iprime}} + \gamma_U\!\!
\sum_{i j, \iprime \jprime}{(\delta_{ij}^U  -\delta_{\iprime \jprime}^U)^2 q_{ij, \iprime\jprime}^U} \notag\\
 &=  (\bdd_A-\bm{\theta}_A)^T\Phi(\bdd_A-\bm{\theta}_A) + (\bdd_U-\bm{\theta}_U)^T\Xi(\bdd_U-\bm{\theta}_U) \notag\\
 & \qquad + \gamma_A \bdd_A \Omega_A \bdd_A + \gamma_U\bdd_U \Omega_U \bdd_U, 
\label{eqn:unit-level-loss-vectorized}
\end{align}
defining $\Omega_A$ and $\Omega_U$ via Lemma \ref{lemma:squared-differences}.

\begin{corollary}
  \label{unit}
  The posterior risk of the loss \eqref{eqn:unit-level-loss-vectorized} is
  minimized by the estimators $\tilde{\bm\theta}^S_A = (I_m + \gamma_A\Phi^{-1}
  \Omega_A)^{-1} \hat{\bm{\theta}}_A^B$ and $\tilde{\bm\theta}^S_U = (I_N +
  \gamma_U\Xi^{-1} \Omega_U)^{-1} \hat{\bm{\theta}}_U^B.$
  \label{thm:unit-level-smoothing}
\end{corollary}

\begin{proof}
  First, note that the ``$m$'' of Theorem~\ref{thm:general-smoothing} is in fact $m+N$ in this setting.
  Now partition $\bm{\theta} =
  (\bm{\theta}^A,\bm{\theta}^U)$.  Similarly, partition the estimate vector as
  $\tilde{\bm\theta}^S = (\tilde{\bm\theta}^S_A, \tilde{\bm\theta}^S_U)$.  Set both
the $\Phi$ and $\Omega$ matrices to be block-diagonal:
\par\vspace*{-\parskip}\noindent
 $$\Phi = \begin{bmatrix} \Phi & 0 \\ 0 & \Xi\end{bmatrix},\qquad
 \Omega =
 \begin{bmatrix}  \Omega_A & 0 \\
   0 &\frac{\gamma_U}{\gamma_A}\Omega_U \end{bmatrix}.$$  Now apply Theorem
 \ref{thm:general-smoothing}.
\end{proof}

\begin{remark}
This device of partitioning was employed by \cite{datta_1991}.
It can be combined with multivariate parameters\footnote{As
before,
group the parameters in
  $\theta^A$ and $\theta^U$ by component.  Then $\Phi$ and $\Xi$ are diagonal;
  $\Omega_A$ and $\Omega_U$ are block-diagonal, each block a copy of the
  univariate $\Omega_A$ or $\Omega_U$.}, and indeed with more than two levels
of spatial hierarchy, if needed.  Since $\Phi$ and $\Omega$ are
block-diagonal, the optimizations over $\tilde{\bm\theta}^S_A$ and
$\tilde{\bm\theta}^S_U$ can be done separately, but no separate theorem is
required.
\end{remark}

\subsection{Benchmarking for Unit-Level Models}

Unit-level models can be benchmarked either for weighted means or for both weighted means and weighted variability.

\subsubsection{Weighted Mean}
\label{sec:multiple-weighted-means}

Consider a unit-level model in which we wish to benchmark both the weighted mean of the area-level estimates and the weighted means of the unit-level estimates within each area.  Then we wish to
minimize~(\ref{eqn:unit-level-loss-vectorized}) under the constraints
\par\vspace*{-\parskip}\noindent
\begin{align}
\sum_{i}{\eta_i \delta_i}  =  t^A, \qquad
\sum_{j}{w_{ij} \hat{\theta}^U_{ij}}  =  \delta_i \quad \forall\; i.
\label{eqn:unit-level-weighted-mean-constraints}
\end{align}
Partition $\bm{\theta}$, $\hat{\bm{\theta}}^B$, and $\bdd$ as in \S
\ref{sec:partitioning}.  Define $\tilde{W}$ as the $m \times N$ matrix such
that\footnote{The $i^{\mathrm{th}}$ row of $\tilde{W}$ will have non-zero
  entries $w_{ij}$ in the columns corresponding to the units in area $i$, and
  zeroes everywhere else.} $(\tilde{W} \bm{\theta}^U)_i =
\sum_{j}{w_{ij}\theta^U_{ij}}$, and define
\par\vspace*{-\parskip}\noindent
$$M = \begin{bmatrix}\bm\eta^T & \bm{0}_N \\
    -I_m & \tilde{W} \end{bmatrix},$$
where $I_m$ is the $m\times m$ identity
matrix and $\bm{0}_{N}$ is the length-$N$ vector of zeroes.  Let $\bm{t} = (t^A,
\bm{0}_m),$ where again $\bm{0}_m$ is the length-$m$ vector of zeroes.  Then
\eqref{eqn:unit-level-weighted-mean-constraints} amounts to
 $M \bdd = \bm{t}$.  By a direct application of Theorem \ref{thm:general-linear-benchmarked}, we have the following result.

\begin{corollary}
The benchmarked Bayes estimator that minimizes the posterior risk in~(\ref{eqn:unit-level-loss-vectorized}) under the constraints in~(\ref{eqn:unit-level-weighted-mean-constraints}) is
  \[
  \tilde{\bm\theta}^S \!=\! \Sigma^{-1} \!\left[ \Phi \hat{\bm{\theta}}^B \!+\!
    M^T \!(M \Sigma^{-1} M^T)^{-1} \!\!\left ( \bm{t} - M \Sigma^{-1} \Phi
      \hat{\bm{\theta}}^B \right ) \!\right]\!,
  \]
  where $\Sigma = \Phi + \gamma \Omega$, and where $\Phi$ and $\Omega$ are as in the
  proof of Corollary \ref{thm:unit-level-smoothing}.
\end{corollary}

\section{Lemma on Squared Differences}

\begin{lemma}
  \label{lemma:squared-differences}
  For a suitable matrix $\Omega$,
\[
    \sum_{i,\iprime}{(\delta_i - \delta_\iprime)^2 q_{i\iprime}} = \bm{\delta}^T \Omega \bm{\delta} .
\]
\end{lemma}

\begin{proof}
  Begin by expanding the square and collecting terms:
\par\vspace*{-\parskip}\noindent  
  \begin{align*}
   & \sum_{i,\iprime}{(\delta_i - \delta_{\iprime})^2 q_{i\iprime}}  \\
    &= \sum_{i,\iprime}{\delta_i^2 q_{i\iprime}} + \sum_{i,\iprime}{\delta_{\iprime}^2 q_{i\iprime}} - 2 \sum_{i,\iprime}{\delta_i \delta_{\iprime}q_{i\iprime}}\\
    & =  \sum_{i}{\delta_i^2 \sum_{\iprime}{q_{i\iprime}}} + \sum_{\iprime}{\delta_{\iprime}^2 \sum_{i}{q_{i\iprime}}} - 2 \sum_{i,\iprime}{\delta_i \delta_{\iprime}q_{i\iprime}}.
  \end{align*}
  Now define the diagonal matrix $Q^{(r)}$ with elements
  $q^{(r)}_{ii} = \sum_{\iprime}{q_{i\iprime}}$, and define the diagonal matrix $Q^{(c)}$ with elements
  $q^{(c)}_{jj} = \sum_{i}{q_{ij}}$.  Substituting,
\par\vspace*{-\parskip}\noindent
  \begin{align*}
    \sum_{i,\iprime}{(\delta_i - \delta_{\iprime})^2 q_{i,\iprime}} & = \bm{\delta}^T Q^{(r)} \bm{\delta} + \bm{\delta}^T Q^{(c)}\bm{\delta} -2 \bm{\delta}^T Q \bm{\delta}\\
    & = \bm{\delta}^T\left(Q^{(r)} + Q^{(c)} -2Q\right) \bm{\delta},
  \end{align*}
  which defines $\Omega$.
\end{proof}

\begin{remark}
In an unweighted, undirected graph with adjacency matrix~$A$, the degree matrix $D$ is defined by $D_{ii} = \sum_{j}{A_{ij}}$, $D_{ij} =
0$; the graph Laplacian in turn is $L=D-A$ \citep{MEJN-on-networks}.  If $Q$ is an
adjacency matrix, then $Q^{(r)} = Q^{(c)} = D$, and $\Omega = 2L$.
\end{remark}

\begin{remark}
By construction, $\Omega$ is clearly positive semi-definite.
It is not positive definite, because $(1\;1\; \cdots\; 1)$ is always an eigenvector,
of eigenvalue zero.  This corresponds to the fact that adding the same constant
to each $\delta_i$ does not change $\sum_{i,\iprime}{(\delta_i -
  \delta_\iprime)^2 q_{i,\iprime}}$.  (These are of course basic properties of
graph Laplacians.)
\end{remark}

\section{Residual Bootstrap}
\label{app:bootstrap}

We consider the model
\par\vspace*{-\parskip}\noindent
\begin{align*}
y_i & = \theta_i + U_i\\
\theta_i &= \bm x_i^T\bm\beta + \epsilon_i
\end{align*}
where $i=1,\ldots,m$ and 
where the observational noise vector $\bm{U}$ has a known diagonal
covariance matrix $\Sigma_U$, with the $i$th diagonal element of $\Sigma_U$ denoted by $\sigma^2_{U,i}$.  We impose constraints (benchmarking, smoothing) on
the estimates of the $\theta_i$ in order to better regularize and borrow
strength.  Call the constrained estimates $\tilde{\bm{\theta}}^{BM}$.  Then we can
define residuals for each observation:
\par\vspace*{-\parskip}\noindent
\begin{align*}
r_i = y_i - \tilde{\theta}^{BM}_i.
\end{align*}
If we standardize these as
\par\vspace*{-\parskip}\noindent
\begin{align*}
\tilde{r}_i = \frac{y_i - \tilde{\theta}^{BM}_i}{\sigma_{U,i}^{}},
\end{align*}
we get quantities which should have the same distribution for all areas, if our
constraints are valid and our model fits well.  We then bootstrap by re-sampling these residuals:  
\par\vspace*{-\parskip}\noindent
\begin{align*}
u_i^\star & \stackrel{\text{iid}}{\sim} \tilde{\bm{r}}\\
y^{\star}_i & = \theta^{BM}_i + u^\star_i\sigma_{U,i}^{}
\end{align*}
where $i=1,\ldots,m$.
Note that the first line of the above model simply means that we draw iid random variables $u_1^\star,\ldots,u_m^\star$ where each $u_i^\star$ is equal to each of $\tilde r_1,\ldots,\tilde r_m$ with probability $1/m$.
Re-sampling--based bootstraps are commonly used in assessing uncertainty for
regression models.  They presume the correctness of the functional form of the
regression, but not of distributional assumptions about the
noise.\footnote{There is also a ``wild bootstrap'' \citep[p.\ 272]{Davison-Hinkley-bootstrap} which would evade having to know the
  observational noise variances, at some cost in efficiency.}

To summarize, the resampling procedure would be this:
\begin{enumerate}
\item From data $(\bm{x}, \bm{y})$, obtain constrained estimates
  $\tilde{\bm{\theta}}^{BM}$ and residuals $\bm r = \bm{y} -
  \tilde{\bm{\theta}}^{BM}$.
\item Calculate standardized residuals $\tilde{\bm{r}} =
  \Sigma^{-1/2}_{U}\bm r$.
\item Repeat $B$ times:
  \begin{enumerate}
  \item Draw $\bm u^\star$ by resampling with replacement from $\tilde{\bm{r}}$.
  \item Set $\bm{y}^\star = \tilde{\bm{\theta}}^{BM} + \Sigma^{-1/2}_{U}\bm{u}^\star$.
  \item Re-run inference on $(\bm{x}, \bm{y}^\star)$ to get
    $\tilde{\bm{\theta}}^{BM\star}$.
  \end{enumerate}
\item Use the distribution of $\tilde{\bm{\theta}}^{BM*}$ in bootstrap calculations.
\end{enumerate}

\end{document}